\documentclass[11pt]{article}

\usepackage[utf8]{inputenc}
\usepackage[T1]{fontenc}
\usepackage{lmodern}
\usepackage[margin=1in]{geometry}
\usepackage{amsmath,amssymb,amsthm}
\usepackage{mathtools}
\usepackage{stmaryrd}
\usepackage{graphicx}
\usepackage{booktabs}
\usepackage{array}
\usepackage{xcolor}
\usepackage{microtype}
\usepackage{enumitem}
\usepackage{caption}
\usepackage{subcaption}
\usepackage[section]{placeins}
\usepackage{algorithm}
\usepackage{algpseudocode}
\algnewcommand\algorithmicparfor{\textbf{parallel for}}
\algnewcommand\algorithmicendparfor{\textbf{end for}}
\algdef{SE}[PARFOR]{ParFor}{EndParFor}[1]%
  {\algorithmicparfor\ #1\ \algorithmicdo}{\algorithmicendparfor}
\usepackage[hidelinks]{hyperref}
\usepackage{cleveref}

\newcommand{\BOX}{\Box}
\newcommand{\DIA}{\Diamond}
\newcommand{\imp}{\rightarrow}
\newcommand{\sat}{\Vdash}
\newcommand{\nsat}{\nVdash}
\newcommand{\frm}{\mathcal{F}}
\newcommand{\mdl}{\mathcal{M}}
\newcommand{\Wset}{W}
\newcommand{\Rel}{R}
\newcommand{\Val}{V}
\newcommand{\tv}[1]{\llbracket #1 \rrbracket}
\newcommand{\suc}{\mathrm{succ}}
\newcommand{\Sub}{\mathrm{Sub}}
\newcommand{\classK}{\mathsf{K}}
\newcommand{\classT}{\mathsf{T}}
\newcommand{\classSfour}{\mathsf{S4}}
\newcommand{\classSfive}{\mathsf{S5}}
\newcommand{\bw}{\mathbin{\&}}
\newcommand{\bcomp}[1]{\overline{#1}}

\newcommand{\NCORPUS}{5{,}621}
\newcommand{\NTOTAL}{5{,}624}
\newcommand{\NRAW}{44{,}678}
\newcommand{\NFRAMESfive}{33{,}554{,}432}
\newcommand{\TOTALEVALS}{1.63\times10^{14}}
\newcommand{\SUSTRATE}{60.2}
\newcommand{\NCERTS}{20{,}990}
\newcommand{\CENSUSMIN}{45}
\newcommand{\CPPSINGLE}{0.0471}
\newcommand{\CPPMULTI}{0.814}
\newcommand{\CPPAVX}{5.43}
\newcommand{\GPUBENCH}{38.33}
\newcommand{\GPUMULTISPEED}{47.1}
\newcommand{\GPUAVXSPEED}{7.06}
\newcommand{\MAXMINCM}{2}
\newcommand{\REFONE}{5{,}086}
\newcommand{\REFTWO}{536}
\newcommand{\KSURV}{2}
\newcommand{\PASSIVEFRONTIER}{4}
\newcommand{\MFRONTIER}{6}
\newcommand{\MIRAGEK}{5}
\newcommand{\NSYNTH}{32}
\newcommand{\NSYNTHCHECK}{128}
\newcommand{\NMINE}{476}
\newcommand{\NMINEVER}{374}
\newcommand{\NMINEUNSEP}{102}

\newcommand{\NEXACTBP}{220}
\newcommand{\NNOCORR}{2{,}606}
\newcommand{\NTRIVIAL}{2{,}798}
\newcommand{\VISRANDRED}{90.9}
\newcommand{\VISDEGRED}{80.5}
\newcommand{\PROGCORPUS}{10{,}000}
\newcommand{\PROGBOOLEAN}{8{,}017}
\newcommand{\PROGUNIVERSE}{5{,}467{,}766}
\newcommand{\PROGSEARCHED}{2{,}691{,}329}
\newcommand{\PROGBUDGET}{1{,}000{,}000}
\newcommand{\INTTHROUGHPUT}{18.87}
\newcommand{\FPTHROUGHPUT}{37.26}
\newcommand{\HBMTHROUGHPUT}{1{,}771}
\newcommand{\SEMUTIL}{98.3}
\newcommand{\ATLASMEM}{51.54}
\newcommand{\SEMBW}{1{,}622}
\newcommand{\SEMENERGY}{10.53}

\theoremstyle{plain}
\newtheorem{theorem}{Theorem}
\newtheorem{proposition}{Proposition}
\newtheorem{lemma}{Lemma}
\newtheorem{corollary}{Corollary}
\theoremstyle{definition}
\newtheorem{definition}{Definition}

\theoremstyle{remark}

\title{\textbf{GPU-Accelerated Search and Certification of Bounded
Indistinguishability in Finite Kripke Semantics}}

\author{
Faruk Alpay\thanks{Corresponding author: \texttt{alpay@lightcap.ai}.}\quad
Bar\i{}\c{s} Ba\c{s}aran\\[3pt]
\small Department of Computer Engineering, Bah\c{c}e\c{s}ehir University, Istanbul, Turkey\\
\small \{faruk.alpay,\,baris.basaran\}@bahcesehir.edu.tr
}
\date{\today}

\begin{document}
\maketitle

\begin{abstract}
The finite model property guarantees finite countermodels for invalid modal
formulas, but its standard bounds are exponential and say little about the model
size at which two apparently equivalent formulas first diverge. We turn both
questions into accelerator-enabled, independently checkable experiments. A set of
worlds is encoded as one integer, reducing $\BOX$ and $\DIA$ to bitmask containment
and intersection; we prove the evaluator correct and fuse formula batches into a
block-reduced CUDA kernel. Matched measurements separate a scalar CPU reference,
a bit-parallel 26-thread OpenMP evaluator, and the fused H100 implementation, and
ablate syntax length, modal depth, variable count, frame filtering, valuation
count, fusion, and atomic contention. Hardware calibration reaches
$\INTTHROUGHPUT$ source-level uint32 TOPS and $\HBMTHROUGHPUT$ GB/s; FP32 is
reported only as a separate device diagnostic. The integer semantic stress
sustains $\SEMUTIL\%$ mean NVML busy time at $\SEMENERGY$ nJ/evaluation, and the
atlas materialisation uses $\ATLASMEM$ GB. For $\NTOTAL$ formulas over
$\classK,\classT,\classSfour,\classSfive$, an exhaustive census through five worlds
performs $\TOTALEVALS$ evaluations in $\CENSUSMIN$ minutes and emits checkable
countermodels; every $\classK$ formula that is refuted fails on at most two worlds,
while non-refutation in the stronger classes remains explicitly bounded. Active
synthesis then finds that $\alpha_2=(\BOX\DIA)^2\top$ and
$\alpha_3=(\BOX\DIA)^3\top$ agree on every frame of at most five worlds yet split
on a certified six-world path. Finally, a progressive semantic representation built from
$\PROGCORPUS$ representatives, including $\PROGBOOLEAN$ Boolean-composed formulas,
ranks a common $\PROGUNIVERSE$-pair universe under PCA, UMAP, spectral, and random
layouts. Each method receives a $\PROGBUDGET$-pair budget; $\PROGSEARCHED$ unique
pairs are stressed on six-, seven-, and eight-world frame families, and every
reported late hit is exhaustively rechecked over all labelled frames through five
worlds before its independent witness is accepted. Raw high-dimensional features
are strongest at tight neighbourhood budgets and PCA-10 is strongest under direct
million-pair ranking; the two-dimensional atlas wins only at the broadest tested
$k$ and is presented as a datashaded candidate-generation view, not as a
standalone discovery engine. Code, data, \NCERTS checked certificates, and
rendering scripts are provided.
\end{abstract}

\section{Introduction}
\label{sec:intro}

A modal formula that is not valid on a class of frames is refuted by a
countermodel: a finite frame in the class, a valuation, and a world at which the
formula is false. For the standard logics, the finite model property (FMP)
guarantees that such a countermodel exists and is finite, and filtration through
the subformulas yields a model with at most $2^{|\Sub(\varphi)|}$
worlds~\cite{blackburn,chagrov}. This bound is worst-case and exponential, and it
is the only generally available estimate of countermodel size. Two natural and,
to our knowledge, largely unmeasured questions are: how large are minimal
countermodels in practice, relative to that bound; and how large must a finite
model be before it can distinguish two formulas that coincide on all smaller
models?

The second question adds a model-size axis to a classical depth-stratified
picture. A $d$-round bisimulation game characterises agreement on formulas of
modal depth at most $d$, the modal analogue of an Ehrenfeucht--Fra\"iss\'e
theorem~\cite{dawarotto}; van Benthem's characterisation and Rosen's finite-model
version identify modal logic with the bisimulation-invariant fragment of first-
order logic~\cite{vanbenthem,rosen}. These results stratify \emph{formula depth}
on a fixed pair of pointed models. They do not record the \emph{model size} of the
smallest counterexample to a fixed formula equivalence. We make that second,
orthogonal quantity the object of study. We call a pair of
formulas that are logically equivalent over all models of at most $k$ worlds but
not over some model of $k+1$ a \emph{$k$-indistinguishable pair} (or, informally, a
semantic mirage), and we ask how large $k$ can be for formulas of bounded syntax.

These questions are decidable but expensive. Modal satisfiability is
PSPACE-complete already for $\classK$~\cite{ladner}, so exhaustive enumeration is
feasible only at small model sizes; the dominant scalable techniques, BDD-based and
bounded symbolic model checking~\cite{bryant,mcmillan}, compress one large state
space rather than sweep many small models against many formulas. Our regime is the
opposite: a complete census of many formulas over all small frames at once. For
this we observe that finite Kripke evaluation is almost entirely bitwise. Encoding
a set of worlds as an integer turns the Boolean connectives into bit operations
and the modal operators into a containment test ($\BOX$) and a non-emptiness test
($\DIA$) between a world's successor mask and the truth set of a subformula
(\Cref{sec:bitmask}). A single kernel then evaluates a whole corpus against all
frames of a fixed size in one pass, without ever forming the (frame $\times$
valuation) truth tensor, which makes the complete universe up to five worlds
tractable on one accelerator (\Cref{sec:engine}).

\paragraph{Contributions and findings.}
\begin{enumerate}[leftmargin=1.5em,itemsep=2pt]
\item A bitmask evaluator with a correctness proof (\Cref{thm:correct}) and an
explicit complexity analysis (\Cref{prop:complexity}); a fused data-parallel
kernel; verifiable certificates (\Cref{prop:cert}); and a bounded decision theorem
from the FMP (\Cref{thm:decision}). All $\NCERTS$ certificates verify.
\item \emph{Minimal countermodels are small in the bounded corpus.} Over formulas
of at most seven nodes and two variables, every $\classK$-refutable formula is refuted on at
most $\MAXMINCM$ worlds (\Cref{prop:frontier}); the standard filtration bound is
empirically far from tight in this regime
(\Cref{cor:loose}).
\item \emph{Active bounded-indistinguishability synthesis.} We define
$k$-indistinguishable pairs, reduce them to biconditional countermodels
(\Cref{prop:mirage}), and synthesize a certified $\MIRAGEK$-indistinguishable pair
of at most seven nodes, separated first at $\MFRONTIER$ worlds
(\Cref{prop:synthmirage}).
\item A matched system study across scalar CPU, bit-parallel/OpenMP CPU, and
fused GPU execution, including launch occupancy estimates, local/register
pressure, output traffic, batching, valuation decoding, atomic contention, and
scale-out failure modes
(\Cref{sec:systemevidence}).
\item A representation-guided candidate-retrieval study over $\PROGCORPUS$
formulas. We compare raw features, PCA, UMAP, spectral, and random representations
on one pair universe using retrieval yield and embedding-fidelity diagnostics;
the two-dimensional views use density aggregation and every accepted candidate is
decided by exhaustive search and the independent verifier
(\Cref{sec:aggregate}).
\item A consistency check via correspondence recovery, reproducing the Sahlqvist
landmarks and isolating McKinsey, consistent with the Goldblatt--Thomason theory of
modal definability~\cite{goldblatt} (\Cref{sec:backproject}).
\end{enumerate}

We are explicit about scope (\Cref{sec:limits}): exhaustive search is a decision
procedure only within the FMP bound, and quantities estimated from sampled frames
beyond the exhaustive range are estimates, while every reported countermodel and
every reported separation is a conclusive, certified witness.

\section{Preliminaries}
\label{sec:prelim}

Modal formulas over variables $\{p_0,p_1,\dots\}$ are
$\varphi ::= \top \mid \bot \mid p_i \mid \neg\varphi \mid \varphi\wedge\varphi
\mid \varphi\vee\varphi \mid \varphi\imp\varphi \mid \BOX\varphi \mid \DIA\varphi$,
with $\Sub(\varphi)$ the subformulas and $|\varphi|$ the node count. A
\emph{frame} $\frm=(\Wset,\Rel)$ has finite $\Wset$ and $\Rel\subseteq\Wset^2$; a
\emph{model} adds a valuation $\Val$. Satisfaction is standard, with
\[
\mdl,w\sat\BOX\varphi \iff \forall v\,(\Rel(w,v)\Rightarrow \mdl,v\sat\varphi),
\quad
\mdl,w\sat\DIA\varphi \iff \exists v\,(\Rel(w,v)\wedge \mdl,v\sat\varphi).
\]
The truth set is $\tv{\varphi}^{\mdl}=\{w:\mdl,w\sat\varphi\}$;
$\frm\models\varphi$ means $\tv{\varphi}^{\mdl}=\Wset$ for all $\Val$; a
\emph{countermodel} is $(\frm,\Val,w)$ with $\mdl,w\nsat\varphi$. We use $\classK$
(all frames), $\classT$ (reflexive), $\classSfour$ (preorders), $\classSfive$
(equivalences), with $\classSfive\subsetneq\classSfour\subsetneq\classT\subsetneq
\classK$, so validity is monotone. Each class has the FMP with a computable bound
$b_{\mathcal C}(\varphi)$: standard or selective filtration gives finite bounds,
including logic-specific constructions for transitive extensions of
$\mathsf{K4}$~\cite{blackburn,chagrov,fine,fine2}. Two pointed models are modally
equivalent iff bisimilar (image-finite case), and $d$-bisimilarity is equivalent
to agreement on modal formulas of depth at most $d$ over a finite
signature~\cite{hennessymilner,dawarotto}. Rosen's finite van Benthem theorem
preserves the bisimulation-invariant characterisation over finite
structures~\cite{rosen}; these are depth/expressiveness results rather than bounds
on the size of a smallest separating model. The modal $\mu$-calculus adds
fixpoints and has its own strict alternation hierarchy~\cite{bradfield}; our
constant-only path family uses no fixpoint and makes no hierarchy claim.
Moreover,
modally definable frame classes are characterised by the Goldblatt--Thomason
theorem~\cite{goldblatt}, and the McKinsey axiom $\BOX\DIA p\imp\DIA\BOX p$ defines
a class that is not first-order definable~\cite{chagrov,vanbenthem}.

\section{A bitmask model of Kripke evaluation}
\label{sec:bitmask}

Fix $n=|\Wset|$ and identify $\Wset$ with $\{0,\dots,n-1\}$. A subset
$S\subseteq\Wset$ is the integer $\widehat S=\sum_{w\in S}2^w$, with
$\mathit{full}=2^n-1$ and $\bcomp x=\mathit{full}\bw(\mathord\sim x)$. A frame is
the array of successor masks $\suc[w]=\widehat{\{v:\Rel(w,v)\}}$; a valuation is one
mask per variable; truth sets are masks computed by the bitwise recurrence on
$\widehat{\tv\cdot}$ with $\widehat{\tv{\neg\varphi}}=\bcomp{\widehat{\tv\varphi}}$,
$\widehat{\tv{\varphi\wedge\psi}}=\widehat{\tv\varphi}\bw\widehat{\tv\psi}$,
$\widehat{\tv{\varphi\vee\psi}}=\widehat{\tv\varphi}\mathbin{|}\widehat{\tv\psi}$,
and the modal cases of \Cref{lem:modal}.

\begin{lemma}[Modal bitmask identities]
\label{lem:modal}
For every model on $n$ worlds and each world $w$,
\begin{align}
w\in\tv{\BOX\varphi} &\iff \suc[w]\bw\bcomp{\widehat{\tv\varphi}}=0
   \iff \suc[w]\subseteq\widehat{\tv\varphi}, \label{eq:box}\\
w\in\tv{\DIA\varphi} &\iff \suc[w]\bw\widehat{\tv\varphi}\neq 0. \label{eq:dia}
\end{align}
\end{lemma}
\begin{proof}
$w\sat\BOX\varphi$ iff every $R$-successor of $w$ satisfies $\varphi$, i.e.\
$\{v:\Rel(w,v)\}\subseteq\tv\varphi$; as bitmasks $S\subseteq T$ iff
$S\bw\bcomp T=0$, giving \eqref{eq:box}. Dually $w\sat\DIA\varphi$ iff
$\{v:\Rel(w,v)\}\cap\tv\varphi\neq\varnothing$, i.e.\ \eqref{eq:dia}.
\end{proof}

\begin{theorem}[Correctness]
\label{thm:correct}
For every formula $\varphi$ and model $\mdl$ on $n$ worlds, the integer computed by
the recurrence equals $\widehat{\tv{\varphi}^{\mdl}}$. Hence $\frm\models\varphi$
iff the computed mask is $\mathit{full}$ under every valuation, and $(\frm,\Val,w)$
is a countermodel iff bit $w$ of the mask is $0$.
\end{theorem}
\begin{proof}
Structural induction on $\varphi$. Atoms and constants hold by the encoding; the
Boolean cases realise the set operations bit by bit using the induction
hypothesis; the modal cases are \Cref{lem:modal} applied to the truth set
$\widehat{\tv\varphi}$ given by the hypothesis. The validity and countermodel
characterisations follow from $\tv{\varphi}^{\mdl}=\Wset$ and from bit $w$ of the
mask being $0$.
\end{proof}

\begin{proposition}[Complexity]
\label{prop:complexity}
Let $\varphi$ have $|\varphi|$ nodes, $d$ of them modal. Evaluating $\varphi$ on one
$n$-world model costs $O(|\varphi|+d\,n)$ word operations. Deciding
$\frm\models\varphi$ over a set $\mathfrak F$ of $n$-world frames and all
valuations of the $k$ variables of $\varphi$ costs
$O(|\mathfrak F|\cdot 2^{kn}\cdot(|\varphi|+d\,n))$ word operations in
$O(|\mathfrak F|\,n)$ memory; no $(\text{frame}\times\text{valuation})$ truth tensor
is stored.
\end{proposition}
\begin{proof}
Each connective is $O(1)$ and each modal operator is the $O(n)$ loop of
\Cref{lem:modal}, giving the per-model bound; multiplying by $|\mathfrak F|$ frames
and $(2^n)^k$ valuations gives the total. Frames occupy $O(n)$ words each and
valuations are enumerated by index, so only the frame array and a constant number
of per-formula reductions are resident.
\end{proof}

The recurrence uses only integer bitwise operations, so it is deterministic across
backends: a CPU (NumPy, reference) and a GPU (CuPy) implementation
compute the same integers, which we verified bit-for-bit on every axiom and frame
size used below. This determinism is what makes the certificates of
\Cref{sec:certificates} trustworthy regardless of where the search ran.

\section{Data-parallel evaluation and measured baselines}
\label{sec:engine}

The evaluation is data-parallel along frames and valuations. Rather than launching
one kernel per operator, we compile the whole corpus into a single kernel
(\Cref{alg:scan}): each thread owns one frame, sweeps every valuation enumerated
from its loop index (so no valuation array is materialised), evaluates each formula
on a small in-register bitmask stack by \Cref{thm:correct}, and reduces with
atomics into per-formula counters. One launch performs
$|\text{corpus}|\times|\mathfrak F|\times 2^{kn}$ evaluations and emits only the
number of falsifying valuations, the first falsifying valuation, and a falsifying
world per formula.

\begin{algorithm}[t]
\caption{Fused corpus scan over one $(\text{class},n)$ layer}
\label{alg:scan}
\begin{algorithmic}[1]
\Require successor masks $\suc[0\mathinner{\ldotp\ldotp}F)[\,0\mathinner{\ldotp\ldotp}n)$; compiled programs of $B$ formulas; $\mathit{full}=2^n-1$
\Ensure $\textsc{fail}[b]$, $\textsc{first}[b]$ per formula
\State $\textsc{fail}[\cdot]\gets 0$;\quad $\textsc{first}[\cdot]\gets\infty$
\ParFor{$f \gets 0$ \textbf{to} $F-1$}
  \State load $\suc[f]$ into registers
  \For{$b \gets 0$ \textbf{to} $B-1$}
    \State $\mathit{lf}\gets 0$
    \For{$v \gets 0$ \textbf{to} $2^{k_b n}-1$}
      \State decode variable masks of $v$; evaluate program $b$ to mask $t$
      \If{$t \neq \mathit{full}$} $\mathit{lf}\gets\mathit{lf}+1$ \EndIf
    \EndFor
    \If{$\mathit{lf}>0$}
      \State $\textsc{atomicAdd}(\textsc{fail}[b],\mathit{lf})$;\;
             $\textsc{atomicMin}(\textsc{first}[b],f)$
    \EndIf
  \EndFor
\EndParFor
\end{algorithmic}
\end{algorithm}

The implementation uses 256 threads per block. Successor masks are loaded once
per frame thread; valuations are decoded by shifts and masks from the loop index,
so neither a valuation array nor a truth tensor is materialised. Formula programs
are concatenated RPN bytecode. The compiler reports 32 registers and 272 bytes of
thread-local storage for both the single-formula and batch kernels; the latter
also uses 4,096 bytes of shared memory for block reduction. These resources permit
eight 256-thread blocks per SM on H100, a theoretical occupancy of 100\%. The
local-storage report is important: the fixed 48-slot semantic stack is not free,
and long programs can generate local-memory traffic even when launch occupancy
remains high.

The reduction is hierarchical. Each thread accumulates one formula's failures
over its valuations, a block reduces 256 frame-local counts in shared memory, and
only the block leader updates the global count and first-frame index. The original
frame-atomic implementation issued up to two atomics per failing
frame--formula pair. The revised kernel issues at most two per block--formula pair,
which directly targets contention rather than hiding it in an aggregate runtime.
Separate formula launches also emit 16 bytes per frame and formula; the fused
reduction emits 16 bytes per formula. At 262,144 frames this is a
262,144-fold output-traffic reduction, before accounting for repeated frame reads.

By \Cref{prop:complexity} the kernel never forms the truth tensor, which at five
worlds over all $\classK$-frames would hold $3.4\times10^{10}$ entries. On one
NVIDIA H100 PCIe (80\,GB, CUDA 12.8) the census of \Cref{sec:frontier} performs
$\TOTALEVALS$ formula evaluations in $\CENSUSMIN$ minutes, an end-to-end sustained
rate of $\SUSTRATE$ billion evaluations per second. We do not label a smaller
microbenchmark as a ``peak'' rate; device calibration and saturation measurements
are separated in \Cref{sec:systemevidence}.

For a direct baseline we implemented the same RPN stack machine in C++17, compiled
with \texttt{-O3 -march=native} and OpenMP, and added a hand-written AVX-512F path
that evaluates 16 valuations per vector. We evaluated
$\BOX(p\imp q)\imp(\DIA p\imp\DIA q)$ on the same prefix of 262,144 labelled
five-world frames and all 1,024 valuations. \Cref{tab:performance} reports measured
wall-clock throughput on the H100 host's Intel Xeon Platinum 8480+; CPU and GPU
checksums agree. The auto-vectorised bitmask/OpenMP row records whatever
\texttt{-march=native} selects; the explicit AVX-512 row is the stronger CPU
baseline. The H100 is $\GPUMULTISPEED\times$ faster than the former and
$\GPUAVXSPEED\times$ faster than the latter on this matched workload.
The higher $\SUSTRATE$ census rate reflects batching many shorter formulas and
should not be compared as if it were the same kernel.

\begin{table}[htbp]
\centering\small
\begin{tabular}{@{}lrrr@{}}
\toprule
Backend & Parallelism & G evaluations/s & H100 speedup \\
\midrule
C++17, \texttt{-O3 -march=native} & 1 thread & $\CPPSINGLE$ & $813.4\times$ \\
C++17 bitmask + OpenMP & 26 threads & $\CPPMULTI$ & $\GPUMULTISPEED\times$ \\
C++17 AVX-512 + OpenMP & 26 threads & $\CPPAVX$ & $\GPUAVXSPEED\times$ \\
CUDA fused evaluator & H100 & $\GPUBENCH$ & $1.0\times$ \\
\bottomrule
\end{tabular}
\caption{Matched-workload implementation baseline: all three regimes execute the
same RPN program over the same 268,435,456 frame--valuation cases and produce the
same checksum. The explicit AVX-512 row prevents the accelerator speedup from
being stated only against an under-optimised scalar CPU path. Frame isomorphisms are intentionally not quotiented out; the
experiment scans adjacency matrices, while formulas are deduplicated separately
by their exact small-frame semantic signatures.}
\label{tab:performance}
\end{table}

BDD and SAT solvers optimise a different workload---symbolic search in one model
or for one formula---so a raw evaluation-rate comparison would be misleading.
Instead, SAT is used as an independent bounded-minimality check for the mirage
family below. The explicit GPU census and symbolic search are complementary, not
interchangeable, regimes~\cite{bryant,mcmillan}.

\section{Certificates and bounded decision}
\label{sec:certificates}

\begin{definition}[Certificate]
A certificate $c=(\varphi,n,\suc,\Val,w,\mathcal C)$ records a formula, world
count, frame as successor masks, valuation, world, and class. It is \emph{valid}
if $\frm\in\mathcal C$ and $\mdl,w\nsat\varphi$.
\end{definition}

\begin{proposition}[Soundness]
\label{prop:cert}
If $c$ is valid then $\varphi$ is not valid on $\mathcal C$, and validity of $c$ is
decidable in $O(|\varphi|\,n)$ by the recursive satisfaction definition,
independently of the evaluator that produced $c$.
\end{proposition}
\begin{proof}
A valid $c$ exhibits a countermodel, so $\frm\not\models\varphi$. The recursive
evaluation of $\mdl,w\sat\varphi$ costs $O(|\varphi|\,n)$ and shares no state with
the parallel search, so acceptance is an independent witness.
\end{proof}

We implement the checker of \Cref{prop:cert} as a separate recursive evaluator over
explicit sets, sharing no code with the bitmask engine; it is the trust anchor of
the study, and every certificate reported below is accepted by it. The separation
between a fast, intricate search and a slow, transparent checker is deliberate: it
lets us believe the GPU results without auditing the GPU.

\begin{theorem}[Bounded decision and minimal countermodels]
\label{thm:decision}
Fix $\mathcal C\in\{\classK,\classT,\classSfour,\classSfive\}$ and $\varphi$, and
suppose the scan examines every frame in $\mathcal C$ of size $\le N$ and every
valuation. (1) If no countermodel is found and $N\ge b_{\mathcal C}(\varphi)$ then
$\varphi$ is valid on $\mathcal C$. (2) Otherwise the least $n\le N$ with a
countermodel is the minimal countermodel size, witnessed by a certificate.
\end{theorem}
\begin{proof}
(1) If $\varphi$ were not valid it would be refuted, by the FMP, on a frame of size
$\le b_{\mathcal C}(\varphi)\le N$, which the scan detects by \Cref{thm:correct};
finding none gives validity. (2) A size-$n$ frame yields a falsifying mask iff a
genuine size-$n$ countermodel exists, so the least such $n$ is minimal and its
frame, valuation, and world form a certificate.
\end{proof}

\section{Minimal countermodels are small}
\label{sec:frontier}

\paragraph{Corpus and census.}
The census corpus is every \emph{constant-free} modal formula of at most seven
nodes over two variables, reduced to one representative per semantic equivalence
class on all frames of $\le3$ worlds, giving $\NCORPUS$ formulas; the named axioms
are added, for $\NTOTAL$ formulas. Constants are admitted by the language and by
the active synthesis of \Cref{sec:mirage}, but excluded here so that syntactic size
is comparable with the original $\NRAW$-formula enumeration. We scan every frame
of $\le 5$ worlds in each class ($\NFRAMESfive$ in
$\classK$; the constrained classes by relational filtering), recording for each
(formula, class) the minimal countermodel size, a certificate, and the failure
density at each size. The census is $\TOTALEVALS$ evaluations in $\CENSUSMIN$
minutes; all $\NCERTS$ certificates verify (\Cref{prop:cert}).
At $N=5$, absence of a countermodel is reported only as ``no countermodel through
five worlds'' unless $N\ge b_{\mathcal C}(\varphi)$ or an external proof applies.
The ancillary bound audit records 429 such bounded-only survivors in $\classT$,
477 in $\classSfour$, and 553 in $\classSfive$; only the propositional tautology
and axiom K need external or within-bound discharge for the two $\classK$
survivors.

\paragraph{The frontier.}
Minimal countermodel sizes are sharply concentrated (\Cref{fig:density}a): in
$\classK$, $\REFONE$ formulas are refuted on a single world and $\REFTWO$ on two;
only $\KSURV$ have no countermodel through five worlds, and both are discharged by
external proofs in \Cref{prop:frontier}; none of the refutable formulas requires
three worlds.

\begin{proposition}[Single-formula frontier]
\label{prop:frontier}
Every $\classK$-refutable formula in the census corpus has a countermodel on at
most $\MAXMINCM$ worlds. The only two formulas not refuted by the exhaustive scan
are the propositional tautology and the normality axiom K, both independently
known to be valid on all Kripke frames.
\end{proposition}

\begin{corollary}[Loose filtration bound on the bounded corpus]
\label{cor:loose}
For the enumerated constant-free formulas of at most seven nodes over at most two
variables, the standard $\classK$ filtration bound
$b_{\classK}(\varphi)=2^{|\Sub(\varphi)|}$ is empirically loose whenever
$\varphi$ is refutable: at seven nodes the bound is up to $2^{7}=128$ worlds,
while the realised maximum is $\MAXMINCM$.
\end{corollary}

\begin{figure}[htbp]
\centering
\includegraphics[width=\textwidth]{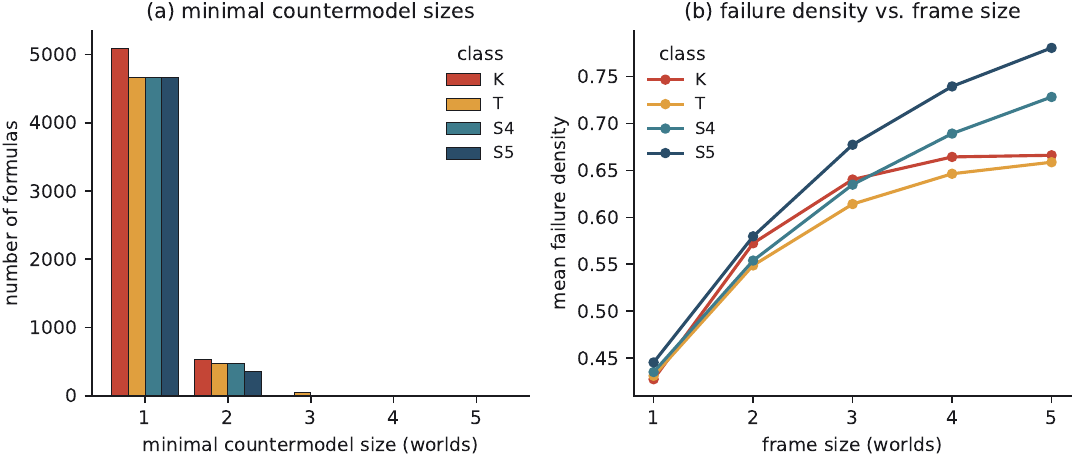}
\caption{Two views of the same small-countermodel phenomenon. (a) The first
failing layer is concentrated at one or two worlds in every logic. (b) Conditional
on refutability, failure mass then grows smoothly with frame size. Together the
panels distinguish \emph{when} a countermodel first appears from \emph{how broadly}
the formula fails once larger frames are available.}
\label{fig:density}
\end{figure}

\Cref{cor:loose} is a corpus-level empirical statement, not a new global
small-model theorem. Standard and selective filtrations, including the transitive
$\mathsf{K4}$ setting, have logic-specific constructions and bounds
\cite{fine,fine2,chagrov}. Within the stated $\le7$-node, $\le2$-variable corpus,
however, the generic exponential bound is a poor predictor: witnessed refutations
live on one or two worlds. The phenomenon is uniform across the four logics
(\Cref{fig:density}a) and across syntactic size: adding modal operators to a
refutable formula does not push its minimal countermodel deeper, because a
falsifying valuation can be exhibited near the root. Two consequences shape the
rest of the paper. Methodologically, every positive minimal-countermodel claim is
exact because all smaller layers were exhausted, while the two negative cases are
settled by standard validity proofs rather than by truncation. Substantively, if
one wants refutations that are genuinely forced to be large, isolated formulas
from this census are the wrong place to look; the natural place is the space of
\emph{biconditionals}, where a large minimal countermodel means precisely that two
formulas agree on every small model and diverge only on a larger one.
\Cref{fig:density}b shows the complementary quantity, mean per-frame failure
density, rising smoothly with frame size as larger frames admit more falsifying
valuations.

\paragraph{Bounded non-refutation between the logics.}
The same census yields the separation counts of \Cref{fig:separation}. Frame-class
inclusion guarantees that a countermodel in a stronger class is also a
countermodel in every weaker class. The off-diagonal entries therefore count
formulas refuted in the row class but with no countermodel through five worlds in
the column class: 428 from $\classK$ to $\classT$, a further 48 to
$\classSfour$, and 118 more to $\classSfive$. They are bounded separation counts,
not theorem counts.

\begin{figure}[htbp]
\centering
\includegraphics[width=\textwidth]{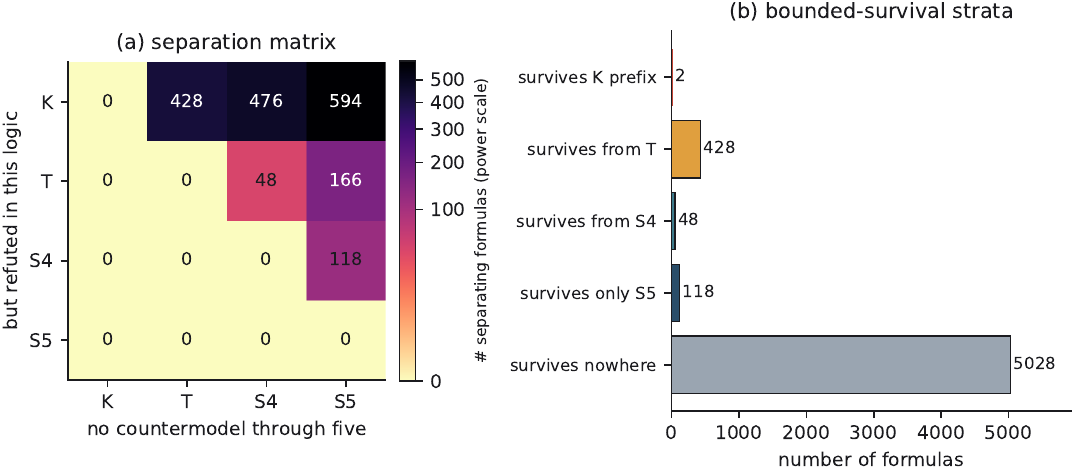}
\caption{(a) Bounded separation counts: formulas with no countermodel through five
worlds in the column class but refuted in the row class. A sublinear power colour
scale preserves the small nonzero strata; cell labels are exact. (b) The corpus
partitioned by the first class in which it survives the five-world scan.}
\label{fig:separation}
\end{figure}

The strata in \Cref{fig:separation}b place each formula in the weakest of the four
classes in which it survives through five worlds. The large ``survives nowhere''
class is expected: a single adversarial valuation usually suffices to refute a
syntactically generated formula on a small frame. The thin upper strata are
bounded survivors, not automatically theorems. These finite-prefix counts are
exact, not sampled, and serve as a baseline for the next section: $\classK$
refutability in this corpus is settled almost immediately, whereas pairwise
indistinguishability can persist to larger models.

\section{Bounded indistinguishability}
\label{sec:mirage}

\begin{definition}[$k$-indistinguishable pair]
\label{def:mirage}
Two formulas $\varphi,\psi$ are \emph{$k$-indistinguishable} if
$\frm\models\varphi\leftrightarrow\psi$ for every frame $\frm$ with $|\Wset|\le k$
but not for some frame with $|\Wset|=k+1$. The least such failing size $k+1$ is
their \emph{separation size}. We also call such a pair a semantic mirage.
\end{definition}

\begin{proposition}[Reduction to biconditional countermodels]
\label{prop:mirage}
$(\varphi,\psi)$ is $k$-indistinguishable iff the biconditional
$\chi=\varphi\leftrightarrow\psi$ has minimal countermodel size exactly $k+1$ in
$\classK$. The separating $(k+1)$-world model is a conclusive witness, verifiable
by \Cref{prop:cert}, and exhibits a world at which $\varphi$ and $\psi$ disagree.
\end{proposition}
\begin{proof}
$\varphi$ and $\psi$ are equivalent on a frame $\frm$ iff $\frm\models\chi$. Hence
equivalence on all frames of size $\le k$ together with failure at $k+1$ is exactly
$\chi$ being valid up to $k$ and refuted at $k+1$, i.e.\ minimal countermodel size
$k+1$ by \Cref{thm:decision}; by \Cref{thm:correct} the falsifying pointed model is
one where $\varphi,\psi$ differ.
\end{proof}

This recasts the question raised in \Cref{sec:intro} as a measurement distinct
from modal-depth equivalence. A $d$-round bisimulation game asks whether two fixed
pointed models agree on all formulas up to depth $d$; \Cref{def:mirage} instead
fixes two formulas and asks for the least \emph{model size} on which they disagree.
Neither axis determines the other. \Cref{prop:mirage} reduces the size axis to the
minimal-countermodel quantity already produced by the census, now applied to
biconditionals. The additional information is operational: it states how far an
exhaustive finite-model search must grow before this particular inequivalence can
be witnessed, without claiming a new bisimulation or expressiveness hierarchy.

\paragraph{Passive census baseline.}
As a diagnostic, we first give each of the $\NRAW$ constant-free formulas of at
most seven nodes a 64-bit per-size behaviour fingerprint and split fingerprint
classes as the frame size grows. This avoids a quadratic pair enumeration and
produces a rightmost split at $\PASSIVEFRONTIER$ worlds. The hash is used only for candidate generation and
visualisation: because exclusive-or hashes can collide or cancel, no equivalence
claim rests on a matching fingerprint. Every pair promoted to a result is retested
by evaluating its biconditional over the complete smaller-frame prefix.

\paragraph{Active synthesis protocol.}
The census can only discover pairs already present in its grammar. We instead
generate formula pairs by iterating modal contexts. For
$C\in\{\BOX\DIA,\DIA\BOX,\BOX\BOX,\DIA\DIA\}$, bases
$b\in\{\top,\bot\}$, and $1\le m\le4$, the candidate is
$(C^m b,C^{m+1}b)$, giving $\NSYNTH$ pairs. We rank candidates by
\[
R(\varphi,\psi)=a(\varphi,\psi)
 +2\,\mathbf 1[\text{separation at }a+1]
 -0.02\bigl(|\varphi|+|\psi|\bigr),
\]
where $a$ is the largest exhaustively verified agreement layer. One fused launch
per size evaluates every candidate biconditional on all
$\sum_{n=1}^{5}2^{n^2}=33{,}620{,}498$ $\classK$-frames through five worlds; the
constant formulas require only one valuation. We perform $\NSYNTHCHECK$ independent
CPU/GPU parity checks through four worlds. Candidates surviving the exact prefix
are challenged on paths, reverse paths, looped paths, and forked paths of six to
ten worlds. A hit is retained only if the recursive checker accepts its separating
model.

We also ran the same objective as an automated mining experiment over all context
words over $\{\BOX,\DIA\}$ of length at most three, bases
$\{\top,\bot,p,p\to p,p\wedge\neg p,\neg p\vee p\}$, and iterations fitting in
21 nodes. This produced $\NMINE$ candidate pairs. The H100 scan exhausted every
labelled $\classK$-frame through five worlds for every biconditional, checked the
GPU verdict against the independent CPU evaluator through four worlds, and then
stressed the survivors on 39 path, reverse-path, reflexive-path, fork, cycle,
star, sparse, and dense templates at each size from six through sixteen. It found
and independently verified $\NMINEVER$ witnesses; $\NMINEUNSEP$ candidates
survived the exact prefix and all stress templates. Under the score
\[
J(\varphi,\psi)=a(\varphi,\psi)+2\,\mathbf 1[\text{template witness after }a]
 -0.02(|\varphi|+|\psi|)+0.1/|C|,
\]
the top exact five-indistinguishable hits were simple path-depth families such as
$\DIA^5\top$ versus $\DIA^6\top$ and $\BOX^5\bot$ versus $\BOX^6\bot$; the
target pair $\alpha_2,\alpha_3$ was also recovered with the six-world path
witness. Thus the automation rediscovers the hand-stated family, but the output
is still bounded evidence: hits beyond six have an exact five-world prefix proof
and a checked larger template witness, not a proof of minimality through all
intermediate labelled frames.

\paragraph{Adversarial frame synthesis.}
The preceding stress test still supplies frames from named families. We therefore
ran a second H100 experiment that fixes $\chi=\varphi\leftrightarrow\psi$ and
searches directly over synthesized relations $R\subseteq W^2$. The selected set
contains 134 exact-prefix candidates and calibration controls from the mining
report: 90 with a known late template witness, 40 that survived all templates, and
four witness-rich controls. For each $n\in\{6,8,10,12,14,16\}$ the generator
produces labelled relations by stratified Erdos--Renyi sampling, edge-flip
mutation of canonical seeds, SCC bottleneck construction, dead-end insertion,
back-edge injection, and random permutation of bottleneck blocks. A synthesized
hit is retained only after the independent recursive checker accepts the emitted
countermodel certificate. One-variable candidates are evaluated through
$n=12$ only; at $n=14,16$ the run deliberately restricts to constant-formula
pairs to keep the valuation budget explicit.

\begin{table}[htbp]
\centering\scriptsize
\begin{tabular}{@{}lrrrr@{}}
\toprule
Generator & Relations & Verified hits & Non-template hits & GPU eval. s \\
\midrule
Erdos--Renyi          & 1.737M & 55 & 53 &  9.71 \\
Edge-flip mutation    & 1.128M & 35 & 25 &  7.69 \\
SCC bottleneck        & 1.509M &  0 &  0 &  8.98 \\
Dead-end insertion    & 1.290M &  0 &  0 &  8.47 \\
Back-edge injection   & 1.575M &  0 &  0 &  9.03 \\
Permuted bottleneck   & 1.755M &  0 &  0 &  9.12 \\
\bottomrule
\end{tabular}
\caption{GPU-guided adversarial frame synthesis for fixed biconditionals. The
run generated 8,995,007 labelled relations and 586,365,732 pair-frame tests
(94.68B valuation-frame evaluations). ``Non-template'' means the accepted
witness has template-similarity below the report threshold; the metric penalizes
path-, cycle-, star-, and extreme-density-like frames.}
\label{tab:adversarialframes}
\end{table}

This changes the interpretation of the late witnesses. The best non-template
examples include formulas whose earlier witness source was a path, but whose new
checked countermodel is an eight-world Erdos--Renyi frame with 13 edges, density
0.203, two strongly connected components, largest SCC size seven, and
template-similarity 0.066. Overall the synthesized panels produced 90 verified
hits, 78 of them non-template; six of the previously template-resistant
candidates were separated, while 34 template-resistant and 10 late-template
candidates remained unseparated by this adversarial distribution. Thus the
experiment broadens the candidate generator beyond named templates, but still
does not prove non-existence: surviving pairs are
adversarial-frame-resistant for the stated generators, sizes, and valuation
budget, not proven equivalent.

\begin{proposition}[Synthesized late separation]
\label{prop:synthmirage}
Let $\alpha_m=(\BOX\DIA)^m\top$. The formulas $\alpha_2$ and $\alpha_3$, of five
and seven nodes respectively, agree on every Kripke frame of at most five worlds
and are separated by a six-world frame. Hence they are
$\MIRAGEK$-indistinguishable and their separation size is exactly $\MFRONTIER$.
\end{proposition}
\begin{proof}
The exhaustive scan evaluates $\alpha_2\leftrightarrow\alpha_3$ on every frame of
sizes one through five and finds no falsifying world; by \Cref{thm:correct}, this
establishes the agreement prefix. For separation, take the directed path
$w_0\Rel w_1\Rel\cdots\Rel w_5$, with $w_5$ a dead end. At $w_0$, $\alpha_2$
holds because $\alpha_1$ holds at $w_2$. However $\alpha_1$ fails at $w_4$
(its only successor is the dead end), so $\alpha_2$ fails at $w_2$, and therefore
$\alpha_3$ fails at $w_0$. The emitted certificate records successor
masks $(2,4,8,16,32,0)$ and falsifying world $w_0$ and is accepted by the
independent recursive checker. Since all smaller sizes were exhausted, six is
minimal.
\end{proof}

\begin{lemma}[Uniform path-separation family]
\label{lem:pathfamily}
For every $m\ge1$, the directed path on $2m+2$ worlds separates
$\alpha_m=(\BOX\DIA)^m\top$ from $\alpha_{m+1}$. More precisely, at a path world
with $r$ edges remaining to the dead end, $\alpha_m$ is false exactly when
$r\in\{1,3,\ldots,2m-1\}$.
\end{lemma}
\begin{proof}
For $m=1$, $\BOX\DIA\top$ is false precisely one edge before the dead end. For the
induction step, $\alpha_{m+1}$ is vacuously true at the dead end, false one edge
before it, and on a world with $r\ge2$ remaining edges has the same truth value as
$\alpha_m$ two edges later. Thus its false-distance set is
$\{1\}\cup\{r+2:r\in\{1,3,\ldots,2m-1\}\}=
\{1,3,\ldots,2m+1\}$. At the root of the $(2m+2)$-world path, $r=2m+1$:
$\alpha_m$ is true and $\alpha_{m+1}$ false.
\end{proof}

The lemma gives a linear-size separating frame for every family member; it is a
modal-depth/path-length construction, not a new logical frontier or hierarchy,
and it does not by itself prove minimality. A separate Tseitin SAT encoding
existentially quantifies a labelled relation, valuation, and falsifying world for
each fixed $n$. CaDiCaL 2.1.3 (seed 20260613) proves UNSAT at every smaller size
for $m=1,2,3$, emits SAT witnesses at sizes $4,6,8$, and the recursive checker
accepts each witness. Every UNSAT result has a shipped DRAT proof independently
accepted by \texttt{drat-trim}; CNFs, proofs, hashes, solver versions, and timings
are in the ancillary package. The $m=2$ case is additionally established by the
complete H100 scan in \Cref{prop:synthmirage}.

\begin{table}[htbp]
\centering\small
\begin{tabular}{@{}rrrrrr@{}}
\toprule
$m$ & UNSAT through & First SAT & Solve s & Check s & Largest DRAT MB \\
\midrule
1 & 3 & 4 & 0.026 & 0.260 & 0.001 \\
2 & 5 & 6 & 0.134 & 0.544 & 0.079 \\
3 & 7 & 8 & 2.841 & 3.533 & 4.638 \\
\bottomrule
\end{tabular}
\caption{Proof-producing symbolic baseline for
$\exists\frm,\Val,w\;((\frm,\Val),w\not\models
\alpha_m\leftrightarrow\alpha_{m+1})$. Looping the fixed-size query from one to
$N$ answers existence through $N$. Solve and proof-check times are totals over
all layers through the first SAT result.}
\label{tab:satbaseline}
\end{table}

The symbolic and explicit routes serve different regimes. SAT proves fixed-formula
minimality through seven worlds in seconds and produces proof objects, where
explicit enumeration is blocked by $2^{n^2}$ relations. The H100 route amortises
frame generation and evaluation across thousands of formulas and certificates at
$n\le5$, where invoking a symbolic solver per formula would discard that batching
advantage. We therefore use SAT for deep, narrow minimality and the GPU enumerator
for broad, shallow census and candidate verification.

\paragraph{Results.}
The active protocol finds two dual exact separations at six worlds; the highest
scoring is the pair of \Cref{prop:synthmirage}. It also finds exact separations at
sizes three, four, and five, while several larger formulas remain unseparated
through the exact five-world prefix. The claim is deliberately one-sided: six is
the largest separation synthesized and certified here, not an upper bound for all
formulas of a given size. \Cref{fig:mirage} connects this logical result to the
representation-guided retrieval benchmark of \Cref{sec:aggregate}.

\begin{figure}[htbp]
\centering
\includegraphics[width=\textwidth]{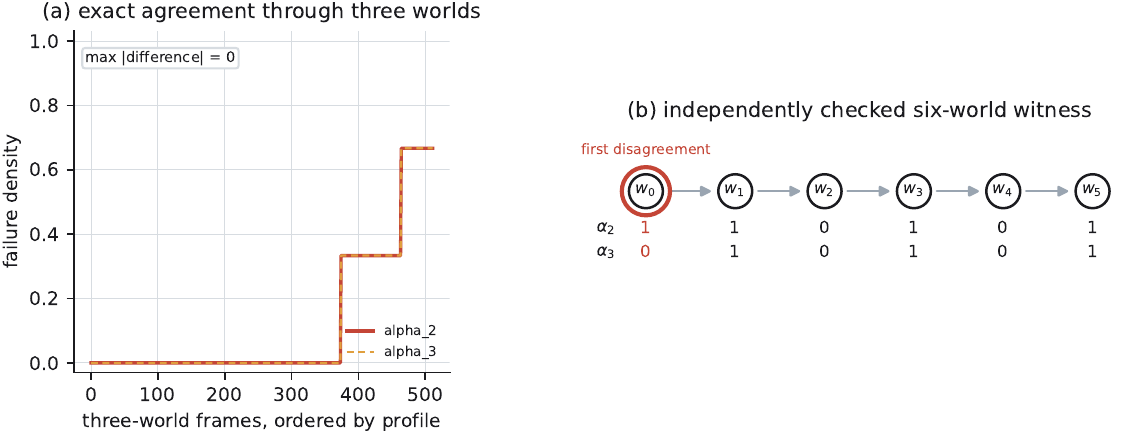}
\caption{Certified semantic mirage, separated from the ranking experiment of
\Cref{fig:progressive}. (a) $\alpha_2$ and $\alpha_3$ have identical failure
profiles on all labelled three-world frames; the exhaustive proof extends this
agreement through five worlds. (b) The independently checked six-world path and
truth vectors; the ringed initial world is the first disagreement.}
\label{fig:mirage}
\end{figure}

The witness exposes why active synthesis matters. The passive grammar contains
$(\BOX\DIA)\top$ only through a three-node encoding of $\top$ as $p\imp p$, so its
seven-node cutoff reaches only the first iteration pair and stops at four worlds.
Admitting the language's primitive constant makes the next pair fit within seven
nodes and moves the exact separation to six. This is not a throughput improvement
over census; it is a change in discovery protocol from asking which existing pairs
split to constructing pairs whose modal alternation delays the split.

\section{Matched system evidence and saturation}
\label{sec:systemevidence}

The baseline of \Cref{tab:performance} is deliberately narrow enough to be
auditable, but it is not sufficient by itself. We therefore use one generic RPN
program in four matched regimes: a textbook world-vector scalar evaluator on one
CPU thread, the bitmask evaluator on 26 OpenMP threads, a hand-written 16-lane
AVX-512/OpenMP evaluator, and the fused H100 kernel.
Every row of \Cref{tab:ablation} uses the same deterministic relation sample,
valuation set, formula, and checksum across regimes. Formula evaluations per
second are reported rather than operator evaluations, so longer formulas expose
their increased work instead of receiving an inflated rate.

\begin{table}[htbp]
\centering\scriptsize
\begin{tabular}{@{}llrrrrrr@{}}
\toprule
Factor & Level & Frames & Valuations & Scalar & Bitmask/OMP & AVX-512/OMP & H100 \\
 & & & & \multicolumn{4}{c}{G formula evaluations/s} \\
\midrule
formula length & 2 ops  & 65,536 & 32 & .0088 & .262 & .300 & 14.87 \\
               & 5 ops  & 65,536 & 32 & .0044 & .223 & .263 & 11.88 \\
               & 11 ops & 65,536 & 32 & .0026 & .237 & 2.589 & 5.50 \\
modal depth    & 1      & 65,536 & 32 & .0088 & .892 & 6.723 & 7.45 \\
               & 4      & 65,536 & 32 & .0047 & .314 & 3.005 & 6.89 \\
               & 8      & 65,536 & 32 & .0032 & .145 & 2.561 & 11.56 \\
variables      & 1      & 65,536 & 32 & .0088 & .867 & 7.112 & 15.72 \\
               & 2      & 8,192  & 1,024 & .0178 & .582 & 8.948 & 7.15 \\
               & 3, underfilled & 512 & 32,768 & .0161 & .732 & 11.447 & .400 \\
valuation count& 3, saturated & 8,192 & 32,768 & .0308 & 1.138 & 15.385 & 6.38 \\
               & 4, underfilled & 32 & 1,048,576 & .0261 & .361 & 5.917 & .0207 \\
frame filter   & $\classK$  & 65,536 & 32 & .0044 & .444 & 4.605 & 7.21 \\
               & $\classT$  & 2,049  & 32 & .0033 & .342 & 2.044 & .240 \\
               & $\mathsf{K4}$ & 304 & 32 & .0022 & .241 & .573 & .0343 \\
\bottomrule
\end{tabular}
\caption{Matched ablation across formula length, modal depth, variable and
valuation count, and frame-class filtering. ``Underfilled'' denotes too few
independent frames to occupy the H100. The scalar and ordinary bitmask columns
show the implementation progression; speedup claims use the explicit AVX-512
column. Table entries are medians over five repeats; the ancillary JSON
contains the timing samples and sample standard deviations. Checksums agree in
every row.}
\label{tab:ablation}
\end{table}

\begin{figure}[htbp]
\centering
\includegraphics[width=\textwidth]{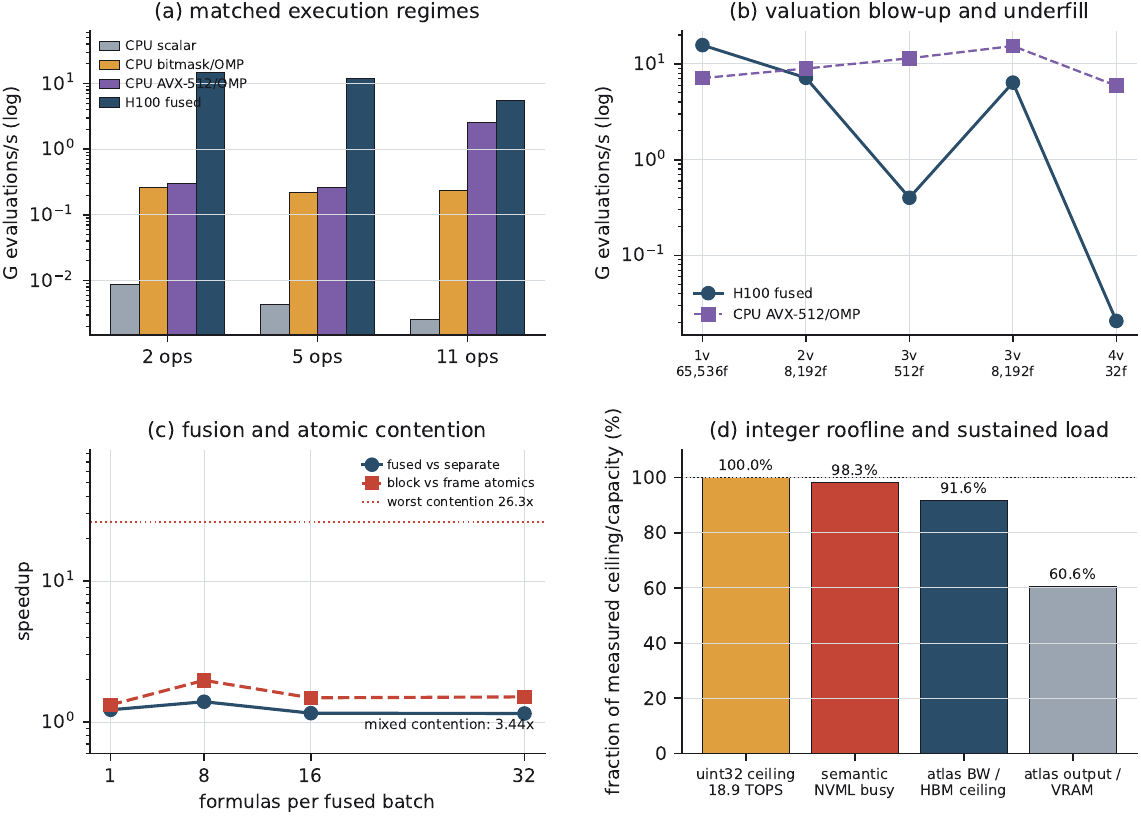}
\caption{System evidence. (a) Matched scalar, bitmask/OpenMP, AVX-512/OpenMP, and fused-GPU
throughput as formula length changes. (b) Variable-count blow-up and the
underfill-to-saturation transition. (c) Fusion reduces launch/output overhead;
block reduction is $26.27\times$ faster than frame atomics under maximum
contention and $3.44\times$ under the mixed workload. (d) Integer and HBM
calibration are separated from application measurements: $\INTTHROUGHPUT$
source-level uint32 TOPS, $\SEMUTIL\%$ mean NVML busy time, $\SEMBW$ GB/s minimum
atlas traffic, and $\ATLASMEM$ GB output.}
\label{fig:system}
\end{figure}

The GPU does not dominate every scale. With three variables and only 512 frames,
the AVX-512 CPU is $28.6\times$ faster than the H100. At 8,192 frames the GPU is
$5.60\times$ faster than ordinary bitmask/OpenMP, but the explicit SIMD baseline
still leads by $2.41\times$. Four variables imply 1,048,576 valuations per
five-world frame; at 32 frames AVX-512 wins by $285.7\times$. The AVX-512 CPU also
wins the filtered $\classT$ and $\mathsf{K4}$ underfill rows. This is the measured
multi-variable and scale-out boundary, not an asymptotic footnote.

Fusion is likewise conditional on its reduction strategy. At batch sizes
1, 8, 16, and 32, block-reduced fusion is respectively $1.22\times$,
$1.40\times$, $1.16\times$, and $1.15\times$ faster than separate launches, while
reducing output bytes by the frame count. The earlier frame-atomic kernel lost up
to $26.27\times$ under adversarial contention and $3.44\times$ on the mixed
workload; shared-memory block reduction removes that failure mode. Finally,
floating-point TFLOPS are not an application metric for an integer/bitwise
semantics kernel. The measured source-level uint32 XOR/AND/OR ceiling is
$\INTTHROUGHPUT$ TOPS; it is a synthetic roofline point, not a dynamic-instruction
efficiency estimate. FP32 reaches $\FPTHROUGHPUT$ TFLOPS only as a separate device
diagnostic. The actual three-variable, depth-ten semantic stress performs 34.36
billion evaluations in 1.427 seconds, or 24.08 G evaluations/s, at $\SEMUTIL\%$
mean NVML busy time. Across ten repeats the elapsed-time standard deviation is
0.000122 s, and mean board power divided by throughput is $\SEMENERGY$
nJ/evaluation. The atlas materialisation writes $\ATLASMEM$ GB and reaches a
conservative $\SEMBW$ GB/s minimum semantic-traffic rate, 91.6\% of the measured
$\HBMTHROUGHPUT$ GB/s HBM ceiling. With Nsight counters blocked by
\texttt{ERR\_NVGPUCTRPERM}, we record RawKernel attributes and 100-ms NVML
power, utilisation, and memory samples.

\FloatBarrier

\section{Representation-guided retrieval and density views}
\label{sec:aggregate}

The atlas matrix is a density-and-ordering view of the same semantic tensor. Let
$D_{ij}\in[0,1]$ be the fraction of valuation--world pairs on which
formula $i$ fails on three-world frame $j$. The $\NCORPUS\times512$ matrix is the
visual encoding: rows are formulas, columns are labelled frames, blue denotes low
failure density, and red high failure density. To place semantically similar frames
next to one another, we centre $D$, take its first principal coordinate, and sort
columns by that coordinate; rows are ordered by the corresponding formula
coordinate. This semantic seriation is deterministic and scales as an SVD of the
display matrix, independent of the much larger census scan.

We evaluate matrix-ordering fidelity rather than relying on visual impression.
This is not a measure of human usefulness. For a column
ordering $\pi$, define adjacent variation
\[
A(\pi)=\frac{1}{B(F-1)}\sum_{i=1}^{B}\sum_{j=1}^{F-1}
  |D_{i,\pi(j+1)}-D_{i,\pi(j)}|.
\]
Here $B$ and $F$ are the numbers of displayed formulas and frames. Lower values
place similar failure profiles together. Across 24 random
permutations, mean $A$ is $0.1360$; ordering by frame out-degree gives $0.06369$;
semantic seriation gives $0.01241$, reductions of $\VISRANDRED\%$ and
$\VISDEGRED\%$ respectively (\Cref{fig:atlas}c). This supplies an explicit
baseline and quantitative criterion for the ordering, but it measures adjacent
profile smoothness rather than human task performance. No perceptual user study
was conducted, a limitation retained in \Cref{sec:limits}.

\begin{figure}[H]
\centering
\includegraphics[width=\textwidth]{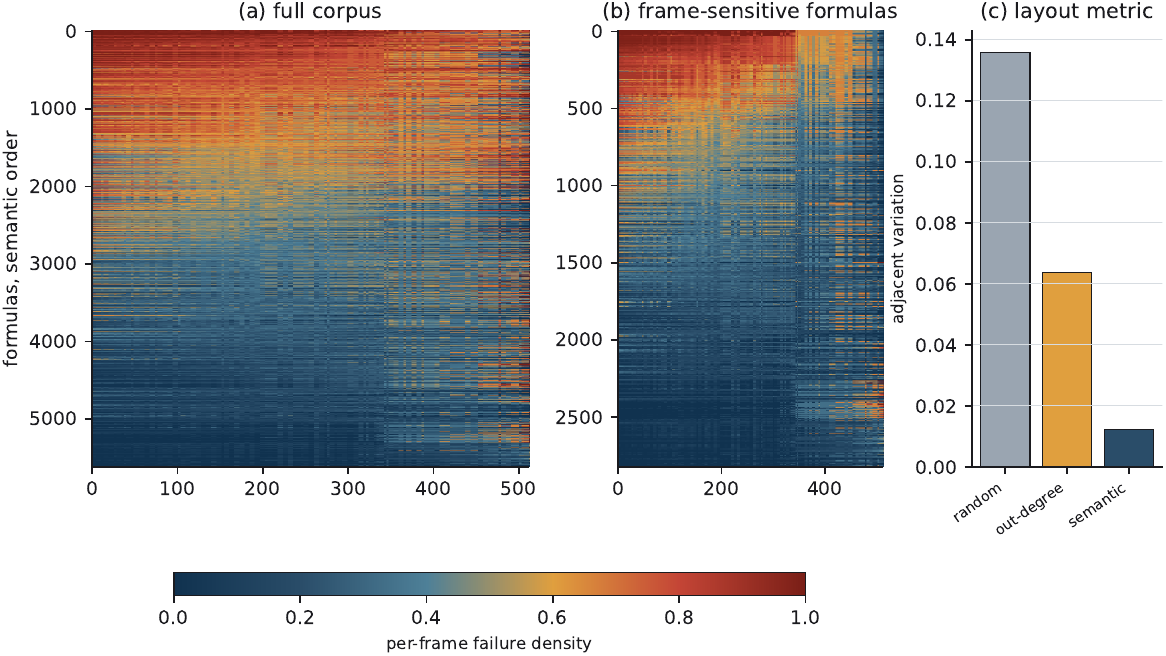}
\caption{Semantic failure atlas over all 512 labelled three-world $\classK$
frames. (a) The full census corpus under semantic seriation. (b) The $2{,}825$
frame-sensitive formulas. (c) Adjacent-column variation for random, out-degree,
and semantic orderings; lower is better.}
\label{fig:atlas}
\end{figure}

\paragraph{Progressive atlas construction.}
The retrieval corpus contains $\PROGCORPUS$ deterministic representatives of
sizes 1--21 and modal depths 0--10; $\PROGBOOLEAN$ contain an explicit Boolean
binary connective, so the experiment is not restricted to unary modal words.
Every layout uses this same formula universe. Atlas-3 concatenates density panels
for all frames through three worlds. Atlas-4 and Atlas-5 append stratified panels
of 2,048 four-world and 4,096 five-world frames, including canonical paths, forks,
loops, cycles, sparse random, and selected dense relations. Exact semantic
screening is separate from layout: two independent 64-bit fingerprints are
computed over every labelled frame and valuation at each size through five.

We compare PCA on Atlas-3, Atlas-4, and Atlas-5 with UMAP and spectral layouts of
the Atlas-5 feature space and a random layout. For each representation we report
15-neighbour purity under the exact source-prefix signature, trustworthiness and
continuity on a fixed semantic sketch, Procrustes correlation after resampling,
the fraction of points sharing a $64\times64$ display bin, and exact $k$-nearest
query latency. The PCA layouts have stability above $0.9988$ and trustworthiness/
continuity above $0.978$, whereas UMAP and spectral stability collapses to $0.095$
and $0.138$. These are dimensionality-reduction fidelity diagnostics, not evidence
of human visual utility. Two-dimensional clutter remains high (98.4--99.98\% for
the semantic layouts), so displayed point fields are density-aggregated rather
than rendered as raw scatter. Querying
all $\PROGCORPUS$ points takes 55--188 ms, or 5.5--18.8 microseconds per point.

\paragraph{Common-universe retrieval task.}
Each layout proposes 256 neighbours per formula; the union defines one
$\PROGUNIVERSE$-pair candidate universe, and every method ranks that same universe
under a $\PROGBUDGET$-pair budget. Across overlapping budgets, $\PROGSEARCHED$
unique pairs are evaluated on 390 frames at each of $n=6,7,8$: canonical path,
reverse path, looped path, fork, cycle, and star frames plus random sparse and
selected dense samples. A reported hit must (i) match the 128-bit fingerprint
through five worlds, (ii) have a six--eight-world distinguishing model, (iii)
survive a complete biconditional scan over every labelled frame and valuation
through five worlds, and (iv) pass the independent recursive checker. The exact
stage rescans all 3,522 screened late candidates in 75.9 seconds; all pass, so the
observed dual-hash collision count is zero.

\begin{figure}[H]
\centering
\includegraphics[width=\textwidth]{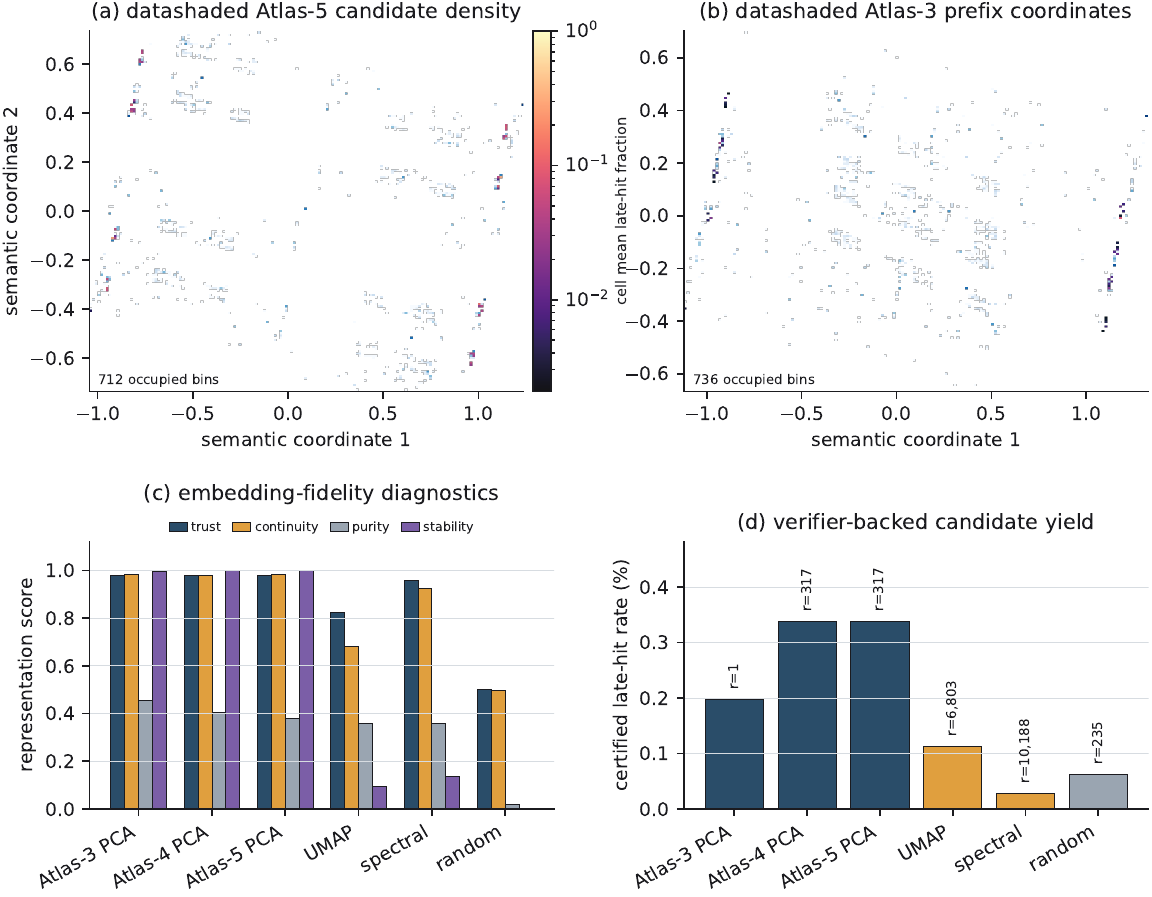}
\caption{Density-aggregated candidate views and retrieval diagnostics. Panels
(a,b) bin formulas in screen space: blue encodes formulas per occupied bin and
magma encodes the bin's mean exact late-hit fraction; no individual-point scatter
is drawn. Panel (c) reports embedding-fidelity diagnostics, not human usefulness.
Panel (d) reports verifier-backed candidate yield under one-million-pair budgets;
labels give the rank of the first verified hit.}
\label{fig:progressive}
\end{figure}

\begin{table}[htbp]
\centering\scriptsize
\begin{tabular}{@{}lrrrrrrr@{}}
\toprule
Layout & Exact hits & Hit rate & First rank & Max. $n$ & Trust./cont. & Stability & Verify ms \\
\midrule
Atlas-3 PCA & 1,970 & 0.1970\% & 1 & 8 & .979/.984 & .9989 & 21.55 \\
Atlas-4 PCA & 3,377 & 0.3377\% & 317 & 8 & .982/.980 & .99995 & 21.55 \\
Atlas-5 PCA & 3,391 & 0.3391\% & 317 & 8 & .979/.986 & .99996 & 21.55 \\
UMAP        & 1,130 & 0.1130\% & 6,803 & 8 & .824/.683 & .095 & 21.55 \\
Spectral    & 283   & 0.0283\% & 10,188 & 8 & .959/.926 & .138 & 21.55 \\
Random      & 627   & 0.0627\% & 235 & 8 & .503/.497 & 0 & 21.55 \\
\bottomrule
\end{tabular}
\caption{Representation-guided candidate-ranking evaluation on one formula universe, pair
universe, and budget. ``Exact hits'' have a complete $n\le5$ prefix proof and an
independently checked $n=6,7,$ or $8$ witness. Verify ms is mean exhaustive-prefix
time per fingerprint-screened candidate; ranking and stress time are additional.}
\label{tab:atlasdiscovery}
\end{table}

Atlas-3, Atlas-4, and Atlas-5 deliver respectively $3.14\times$, $5.39\times$,
and $5.41\times$ the random exact-hit rate; UMAP reaches $1.80\times$, while the
unstable spectral layout falls to $0.45\times$. Atlas-5 yields 1,215 six-world,
1,096 seven-world, and 1,080 eight-world exact hits. These results show that
small-frame semantic representations can prioritize verifier work better than
random ranking on this candidate universe. They do not establish that a person
can discover the pairs from a two-dimensional picture.

\paragraph{Does the picture help?}
To separate two-dimensional layout from representation quality, we reran the
candidate-retrieval task on the same $\PROGUNIVERSE$-pair universe and
one-million-pair budget, but
ranked it with raw Atlas-5 features, PCA-50, PCA-10, the two-dimensional PCA
atlas, UMAP-2, and random coordinates. For each method and
$k\in\{16,32,64,128,256,512,1024\}$, candidate pairs are the $k$ nearest
neighbours per formula, intersected with the fixed universe and then truncated by
distance to the common budget. The certified-hit set is the same 3,522 pairs
passing the complete $n\le5$ biconditional rescan; bootstrap intervals are over
candidate pairs for precision and over certified hits for recall. Full results
are in the ancillary \texttt{picture\_help\_h100.json}; endpoint rows are shown in
\Cref{tab:picturehelp}.

\begin{table}[htbp]
\centering\scriptsize
\begin{tabular}{@{}lrrrrrr@{}}
\toprule
Representation & \multicolumn{3}{c}{$k=16$} & \multicolumn{3}{c}{$k=1024$} \\
\cmidrule(lr){2-4}\cmidrule(l){5-7}
 & Hits & Prec. & Recall & Hits & Prec. & Recall \\
\midrule
Raw high-D   & 3,519 & 2.687\% & 99.91\% & 3,074 & 0.307\% & 87.28\% \\
PCA-50       &   993 & 0.802\% & 28.19\% &   553 & 0.055\% & 15.70\% \\
PCA-10       &   987 & 0.859\% & 28.02\% & 3,166 & 0.317\% & 89.89\% \\
PCA-2 atlas  &   995 & 0.872\% & 28.25\% & 3,242 & 0.324\% & 92.05\% \\
UMAP-2       &   921 & 1.007\% & 26.15\% & 1,426 & 0.143\% & 40.49\% \\
Random       &     5 & 0.006\% &  0.14\% &   147 & 0.015\% &  4.17\% \\
\bottomrule
\end{tabular}
\caption{``Does the picture help?'' endpoint comparison on the fixed
$5{,}467{,}766$-pair universe. Raw high-dimensional features are best at the
tightest neighbourhood ($k=16$; precision 95\% CI $[2.55,2.82]\%$, recall
$[99.80,100]\%$), while the two-dimensional PCA atlas is best at $k=1024$
(precision CI $[0.276,0.374]\%$, recall CI $[91.11,92.85]\%$).}
\label{tab:picturehelp}
\end{table}

The raw high-dimensional representation is the sharpest retrieval signal at very
small $k$, while the two-dimensional PCA atlas retains enough of that signal to
win the broad $k=1024$ ranking. The supported claim is representation-guided
information retrieval followed by exact verification; the datashaded atlas is a
candidate-generation interface, not evidence that visual inspection causes the
discoveries.

\paragraph{Embedding benchmark.}
We also score standard embeddings on the same fixed universe and certified-hit
set, now ranking every eligible pair directly by distance in the embedding space
and taking the nearest one million pairs. The benchmark records task yield,
trustworthiness/continuity, resampling stability where deterministic replicas are
available, and a coarse clutter statistic: the fraction of formula points sharing
a $64\times64$ screen bin. Full rows are in the ancillary
\texttt{embedding\_benchmark\_h100.json}.

\begin{table}[htbp]
\centering\scriptsize
\begin{tabular}{@{}lrrrrr@{}}
\toprule
Embedding & Hits & Recall & Trust./cont. & Stability & Clutter \\
\midrule
PCA-10           & 3,337 & 94.75\% & .984/.984 & .99999 & 98.61\% \\
Parametric UMAP-2& 1,548 & 43.95\% & .982/.983 & --     & 99.99\% \\
UMAP-2           & 1,426 & 40.49\% & .819/.686 & .1109  & 99.95\% \\
PHATE-2          & 1,353 & 38.42\% & .925/.952 & --     & 99.30\% \\
PCA-2 atlas      &   728 & 20.67\% & .979/.986 & .99999 & 98.61\% \\
PCA-50           &   673 & 19.11\% & .985/.982 & .99999 & 98.61\% \\
Autoencoder-2    &   664 & 18.85\% & .779/.783 & --     & 93.42\% \\
TriMap-2         &   402 & 11.41\% & .946/.862 & --     & 95.60\% \\
Spectral-2       &   285 &  8.09\% & .954/.930 & --     & 99.98\% \\
\bottomrule
\end{tabular}
\caption{Embedding benchmark on the same $\PROGUNIVERSE$-pair universe and
$\PROGBUDGET$-pair budget as \Cref{tab:picturehelp}. Yield is measured against
the 3,522 exact-prefix-certified late hits. PCA-10 is the strongest retrieval
representation; among two-dimensional embeddings, parametric UMAP gives the
highest yield but is visually cluttered under this screen-bin metric.}
\label{tab:embeddingbench}
\end{table}

The benchmark strengthens the conservative reading of the atlas. The best pure
retrieval representation is ten-dimensional PCA, not a two-dimensional picture.
Two-dimensional layouts still carry useful retrieval signal, but their value is
as candidate-generation interfaces whose proposals are certified by the verifier,
not as standalone evidence.

\paragraph{Scaling the picture.}
We then scaled the display problem by $10\times$, from 10,000 to 100,000
deterministically generated formulas, and rebuilt Atlas-5 panel features on the
H100. The resulting scatter plot is not a readable figure: on a
$1024\times1024$ screen grid, only 2,782 bins are occupied, 98.61\% of points
share a bin, the largest bin contains 21,171 formulas, and a 1,000-label sample
has 98.0\% rectangular label overlap. Rendering each point as a radius-two mark
would issue 2.5M point-pixel draw operations but cover only 39,222 pixels, an
overdraw factor of $63.7\times$. Density aggregation therefore becomes a
requirement, not a stylistic choice: datashading emits one aggregate per occupied
bin, a $35.9\times$ draw-op reduction. The same run also tested four candidate
policies under the H100 verifier: PCA-10 nearest neighbours, dense screen cells,
random pairs, and an alternating-family anchor control. Among 2,000 proposed
pairs each, the first three produced larger-template disagreements but zero
exact-prefix-certified late hits; the eight-pair anchor control produced three.
The conclusion is deliberately negative and useful: at this scale, visual density
alone is not a discovery engine. The successful pipeline needs a structural or
representation prior to propose candidates, and the H100 verifier remains the
source of ground truth.

\begin{table}[htbp]
\centering\scriptsize
\begin{tabular}{@{}lrrrr@{}}
\toprule
Candidate policy & Pairs & Template witnesses & Certified hits & Hits/GPU-hour \\
\midrule
PCA-10 nearest neighbours & 2,000 & 833 & 0 & 0 \\
Dense datashade cells     & 2,000 & 320 & 0 & 0 \\
Random pairs              & 2,000 & 1,864 & 0 & 0 \\
Alternating anchor control&     8 & 6 & 3 & $2.63\times10^5$ \\
\bottomrule
\end{tabular}
\caption{Ten-times corpus-scale datashading audit. Template witnesses are
larger-frame disagreements on the stress panels; certified hits additionally pass
the complete $n\le5$ biconditional rescan and independent witness checker.}
\label{tab:datashading}
\end{table}

\FloatBarrier

\section{Consistency check: recovering frame correspondents}
\label{sec:backproject}

As a check that the engine recovers established structure rather than artefacts,
we reconstruct first-order frame conditions from refutation data. Over all frames
of size $\le N$ write $\mathrm{Val}_{\le N}(\varphi)=\{\frm:\frm\models\varphi\}$,
and for each condition $X$ in a fixed library write
$\mathrm{Ext}_{\le N}(X)=\{\frm:\frm\models X\}$. The recovered condition of
$\varphi$ is an $X$ with equal extensions.

\begin{proposition}[Soundness over the searched universe]
\label{prop:bp}
If $X$ is recovered for $\varphi$ over $N$, then $\frm\models\varphi\iff\frm\models
X$ for all $\frm$ of size $\le N$. If $\varphi$ is Sahlqvist with first-order
correspondent $X^{\dagger}$, then $X$ and $X^{\dagger}$ agree over all frames of
size $\le N$.
\end{proposition}
\begin{proof}
The first claim is the definition; the second combines it with the Sahlqvist
equivalence over those frames~\cite{sahlqvist}.
\end{proof}

\begin{table}[htbp]
\centering\small
\begin{tabular}{@{}llll@{}}
\toprule
Axiom & Formula & Recovered condition & Recovered \\
\midrule
T  & $\BOX p\imp p$              & reflexive                                   & yes \\
D  & $\BOX p\imp\DIA p$          & serial                                      & yes \\
4  & $\BOX p\imp\BOX\BOX p$      & transitive                                  & yes \\
5  & $\DIA p\imp\BOX\DIA p$      & Euclidean                                   & yes \\
B  & $p\imp\BOX\DIA p$           & symmetric                                   & yes \\
CD & $\DIA p\imp\BOX p$          & partially functional                        & yes \\
C4 & $\BOX\BOX p\imp\BOX p$      & dense                                       & yes \\
G  & $\DIA\BOX p\imp\BOX\DIA p$  & directed (confluent)                        & yes \\
M  & $\BOX\DIA p\imp\DIA\BOX p$  & \emph{none in the library}                  & \textbf{no} \\
\bottomrule
\end{tabular}
\caption{Conditions recovered from countermodel data over all frames of $\le4$
worlds. Every Sahlqvist landmark is recovered; McKinsey is the named control with
no match in the tested library.}
\label{tab:backproj}
\end{table}

The procedure reconstructs every textbook correspondent and returns no tested
condition for McKinsey, consistent with its non-first-order-definable frame class
but not a new proof of that fact~\cite{goldblatt,chagrov,vanbenthem}. Over the
whole corpus it recovers a condition for $\NEXACTBP$ formulas, each re-confirmed
over all $\NFRAMESfive$ five-world frames; $\NTRIVIAL$ formulas have constant
extensions over the searched universe, and $\NNOCORR$ remain unmatched. This is a control, not a
contribution: where theory gives the answer, the recovered extension agrees.

\begin{figure}[H]
\centering
\includegraphics[width=\textwidth]{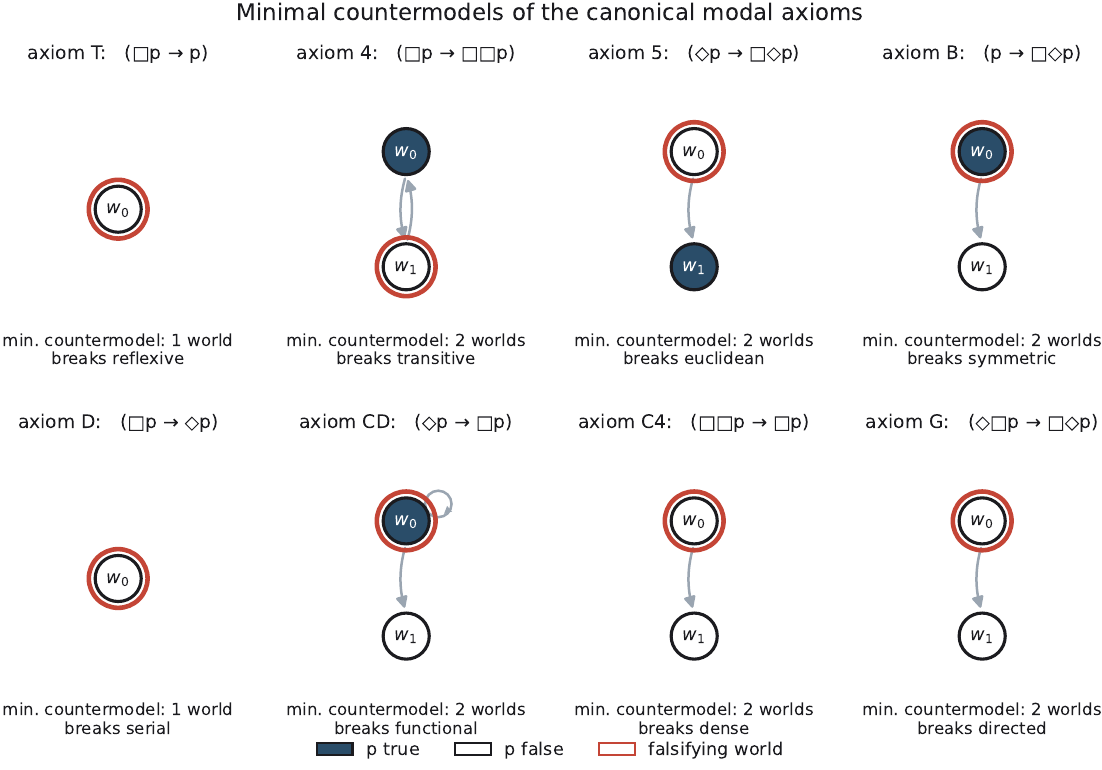}
\caption{Minimal countermodels of the canonical modal axioms. Node fill marks
where $p$ holds; the ringed world is where the axiom fails; the caption of each
panel states the frame condition the model violates.}
\label{fig:gallery}
\end{figure}

\FloatBarrier

\section{Discussion and scope}
\label{sec:limits}

Search is bounded, and we report accordingly. By \Cref{thm:decision} the census
decides validity exactly for formulas whose FMP bound is within the searched size;
the single-formula $\classK$ statement of \Cref{prop:frontier} is exact because every
$\classK$-refutable census formula has a witnessed failure by size two and the two
$\classK$ survivors are independently valid. Passive fingerprint matches are heuristic candidates,
not proofs. By contrast, \Cref{prop:synthmirage} is exact: every smaller frame is
enumerated and the six-world separation is certified. It establishes existence of
a separation at six, not a global upper bound or a complete frontier for all
seven-node formulas.

\paragraph{Failure-mode audit.}
The implementation has explicit finite limits. Truth and successor masks are
32-bit integers, the CUDA successor array has 16 slots, formula programs have a
48-slot stack, and at most four propositional variables are decoded into registers;
the 64-bit valuation index additionally requires $kn<63$. These are implementation
limits, not logical bounds. Exhaustive $\classK$ enumeration grows as $2^{n^2}$:
$n=5$ has 33,554,432 labelled frames, while $n=6$ already has $2^{36}$, so the
complete census stops before the six-world stress families. Variable count is a
second exponential: valuations grow as $2^{kn}$. The measured four-variable,
five-world underfilled case is CPU-faster despite the H100, and storing a full
truth tensor is out of scale even when fused reductions remain small.

Fingerprints are never treated as proofs. The progressive experiment uses two
independent 64-bit set hashes to screen candidates, then exhaustively scans the
biconditional of every reported late hit over every labelled frame and valuation
through five worlds. Thus a hash collision can increase verification work but
cannot create a reported result. Beyond five worlds, path, fork, looped-path,
cycle, random sparse, and selected dense families find witnesses but do not prove
the absence of others; the adversarial-frame synthesis panels enlarge this
challenge set but remain a sampled search distribution, not an exhaustive
enumeration of $2^{n^2}$ relations. Adjacent variation, trustworthiness,
continuity, and Procrustes stability measure ordering or embedding fidelity, not
human visual usefulness. The 98.4--99.98\% collision rates and 98.61\% collapse at
$10\times$ scale rule out a readable raw scatter; the figures therefore use
density aggregation. No controlled user/task study was conducted, so the paper
claims only representation-guided candidate retrieval plus verification. Atlas-4
and Atlas-5 use stratified density panels for layout, while exact-prefix
verification still uses all frames.
Finally, the object language has one accessibility relation; multiple modalities
and relation interaction remain outside the present evidence.

\section{Related work}
\label{sec:related}

Finite Kripke model checking, tableaux, and resolution for modal logic are
classical~\cite{blackburn,chagrov}, and the FMP and its filtration bounds are
standard~\cite{blackburn,fine}; our contribution is not a decision procedure but
an exhaustive semantic experiment plus an active synthesis protocol. The semantic
questions we quantify descend from correspondence and definability theory:
Sahlqvist's theorem~\cite{sahlqvist}, van Benthem's modal-classical
correspondence~\cite{vanbenthem}, and the Goldblatt--Thomason characterisation of
modally definable frame classes~\cite{goldblatt} together explain why
back-projection succeeds on Sahlqvist axioms and why McKinsey has no global
first-order frame correspondent. Our
notion of bounded indistinguishability is complementary to modal-depth
stratification by finite bisimulation games~\cite{dawarotto}: it fixes formulas and
measures the size of the least separating model. Rosen's finite van Benthem
theorem preserves the bisimulation-invariant characterisation on finite
structures~\cite{rosen}, while the modal $\mu$-calculus alternation hierarchy
classifies fixpoint expressiveness~\cite{bradfield}; neither is a model-size
frontier of the kind measured here. Computationally, the PSPACE-completeness of modal
satisfiability~\cite{ladner} bounds what is feasible in general and motivates the
small-model regime; symbolic model checking compresses one large state space with
BDDs~\cite{bryant,mcmillan}, whereas we evaluate many small models against many
formulas explicitly and in parallel, a complementary regime. In information visualisation, seriation orders
matrix rows or columns to reveal structure~\cite{hahsler}; reorderable matrices
and Bertin-style tabular views make this an interactive visual-analysis
operation~\cite{perin}. Our atlas supplies a domain-specific semantic matrix and
an explicit continuity metric for comparing such orderings. The retrieval
experiment treats semantic proximity as a proposal mechanism and evaluates it by
exact-prefix downstream yield; the density view is not evaluated as a human
discovery interface.

\section{Conclusion}
\label{sec:conclusion}

Exhaustive finite-model evaluation, made cheap by bitmask semantics and one
accelerator, supports both census and adversarial discovery. The census shows that
ordinary small formulas fail almost immediately. Active synthesis then finds what
passive enumeration misses: $(\BOX\DIA)^2\top$ and $(\BOX\DIA)^3\top$ are exactly
$\MIRAGEK$-indistinguishable and split on a certified six-world path. The
representation study then uses small-frame semantic geometry to rank candidates:
raw features, PCA, UMAP, spectral, and random representations rank one common pair
universe, and every reported six--eight-world hit survives an exhaustive
five-world prefix scan and an independent witness check. Adversarial frame synthesis removes the named-template
assumption for selected biconditionals, producing 90 checked separations, including
78 non-template witnesses, while reporting the remaining survivors only as
resistant to the stated search distribution. Correspondence recovery remains a
control against known logic. The main result is therefore not raw H100 throughput;
it is a reproducible protocol in which the accelerator proposes or eliminates
extremal objects and a small independent checker certifies every positive witness.

\paragraph{Reproducibility.}
The evaluator, fused kernels, active-synthesis, bounded-indistinguishability
mining, adversarial frame synthesis, progressive-atlas and datashading
experiments, SAT family checks, measured scalar and C++/OpenMP baselines, system
ablations, hardware calibration, embedding-fidelity metrics, figure scripts, census data, and
every certificate are provided, together with the independent checker and the
tests that reproduce the Sahlqvist correspondences.

\end{document}